\documentclass[11pt]{amsart}

\usepackage[margin=1in]{geometry}
\usepackage{amsmath,amssymb,amsthm}
\usepackage{amsfonts}
\usepackage{enumitem,bm,graphicx,hyperref}
\usepackage[sort]{cite}
\newcommand{\R}{\mathbb{R}}
\newcommand{\C}{\mathbb{C}}
\newcommand{\Z}{\mathbb{Z}}

\newcommand{\cI}{\mathcal{I}}

\newcommand{\cC}{\mathcal{C}}
\newcommand{\cA}{\mathcal{A}}

\newcommand{\fS}{\mathfrak{S}}
\newcommand{\bra}[1]{\la #1|}
\newcommand{\ket}[1]{| #1 \ra}

\renewcommand{\vec}[1]{\boldsymbol{#1}}

\newcommand\numberthis{\addtocounter{equation}{1}\tag{\theequation}}

\newcommand{\la}{\langle}
\newcommand{\ra}{\rangle}
\DeclareMathOperator{\range}{range}
\DeclareMathOperator{\tr}{tr}

\DeclareMathOperator{\supp}{supp}
\DeclareMathOperator{\rank}{rank}

\newcommand{\dd}{\,\ensuremath{\textrm{d}}}

\newcommand{\comment}[1]{}

\usepackage{xcolor}
\definecolor{purp}{RGB}{160, 32, 240}

\newtheorem{theorem}{Theorem}
\newtheorem*{conjecture*}{Conjecture}
\newtheorem{lemma}{Lemma}[section]
\newtheorem*{lemma*}{Lemma}
\newtheorem{proposition}[lemma]{Proposition}
\newtheorem*{proposition*}{Proposition}

\theoremstyle{definition}

\newtheorem{definition}{Definition}

\allowdisplaybreaks

\title[Algebraic localization implies Chern triviality in non-periodic insulators]{Algebraic localization of Wannier functions implies Chern triviality in non-periodic insulators}
\author{Jianfeng Lu}
\address{(JL) Department of Mathematics, Department of Physics, and Department of Chemistry, Duke University, Box 90320, Durham, NC 27708, USA}
\email{jianfeng@math.duke.edu}

\author{Kevin D. Stubbs}
\address{(KDS) Department of Mathematics, Duke University, Box 90320, Durham, NC 27708, USA}
\email{kstubbs@math.duke.edu}

\date{\today}
\thanks{This work is supported in part by the U.S.~National Science Foundation via grant DMS-2012286 and the U.S.~Department of Energy via grant DE-SC0019449. K.D.S. was supported in part by a National Science Foundation Graduate Research Fellowship under Grant No. DGE-1644868.}

\begin{document}
\begin{abstract}
For gapped periodic systems (insulators), it has been established that the insulator is topologically trivial (i.e., its Chern number is equal to $0$) if and only if its Fermi projector admits an orthogonal basis with finite second moment (i.e., all basis elements satisfy $\int |\vec{x}|^2 |w(\vec{x})|^2 \dd{\vec{x}} < \infty$). In this paper, we extend one direction of this result to non-periodic gapped systems. In particular, we show that the existence of an orthogonal basis with slightly more decay ($\int |\vec{x}|^{2+\epsilon} |w(\vec{x})|^2 \dd{\vec{x}} < \infty$ for any $\epsilon > 0$) is a sufficient condition to conclude that the Chern marker, the natural generalization of the Chern number, vanishes.
\end{abstract}

\maketitle

\section{Introduction}
In electron structure theory, we are often interested in studying the subspace of low energy states spanned by the range of Fermi projector $P$. For numerical and theoretical purposes, we are in particular interested in finding a basis for the occupied space $\range{(P)}$ which is as well localized in space as possible. The elements of such a basis are known as Wannier functions or generalized Wannier functions (see review \cite{2012MarzariMostofiYatesSouzaVanderbilt} and references therein). Typically for insulating materials, the Fermi projector $P$ admits an integral kernel which is exponentially localized in the following sense (see, for example \cite{2020MarcelliMoscolariPanati}):
\begin{equation}
\label{eq:intro-p}
|P(\vec{x},\vec{y})| \lesssim e^{-c_{gap} |\vec{x} - \vec{y}|}.
\end{equation}
Therefore, we might expect that these insulators admit a basis which decays exponentially quickly in space. Somewhat surprisingly, even if $P$ satisfies an estimate like Equation \eqref{eq:intro-p}, it is not necessarily true that $\range{(P)}$ admits a basis which decays exponentially quickly in space due to the existence of so called ``topological obstructions''.

In two dimensional periodic insulators, it is now well understood \cite{2007BrouderPanatiCalandraMourougane, 2007Panati, 2018MonacoPanatiPisanteTeufel}  that the existence of a well localized basis for $\range{(P)}$ is fully characterized by the Chern number which is defined as follows:
\begin{equation*}
    c(P) = \frac{1}{2\pi} \int_{\mathcal{B}} \tr\Bigl( P(\vec{k}) \bigl[\partial_{k_1} P(\vec{k}), \partial_{k_2} P(\vec{k}) \bigr] \Bigr) d k_1 \wedge d k_2,
\end{equation*}
where $\mathcal{B}$ is the first Brillouin zone and $P(\vec{k})$ is the Bloch decomposition of $P$ (see e.g., \cite{ReedSimoniv}).

For periodic systems, $P$ possesses a basis with finite second moment (known as Wannier functions) if and only if $c(P) = 0$, as established in  \cite{2018MonacoPanatiPisanteTeufel}. Furthermore, $c(P) = 0$ if and only if there exists a basis of $\range{(P)}$ which is exponentially localized \cite{2007BrouderPanatiCalandraMourougane}. These results, which connect the existence of a basis with finite second moment to the vanishing of the Chern number and to the existence of an exponentially localized basis, is known as the localization dichotomy in periodic insulators.

Since the notion of the Chern number depends on the Bloch decomposition, the Chern number is no longer well defined for non-periodic systems. For generic systems, the Chern marker was proposed in \cite{2018CorneanMonacoMoscolari,2019MarcelliMonacoMoscolariPanati} as an extension.
\begin{definition}[Chern Marker]
Let $P$ be a projection on $L^2(\R^2)$ and $\chi_L$ be the indicator function of the set $[-L,L)^2$.  The \textbf{Chern marker} of $P$ is defined by
  \[
    C(P) := \lim_{L \rightarrow \infty} \frac{2 \pi i}{4L^2} \tr{\left( \chi_{L} P \Big[ [ X, P ], [ Y, P ] \Big] P \chi_{L}\right)}
  \]
  whenever the limit on the right hand side exists.
\end{definition}
Note that this generalizes the Chern number as for periodic systems the Chern number and the Chern marker agree \cite{2019MarcelliMonacoMoscolariPanati,2020MarcelliMoscolariPanati}. Therefore, parallel to the periodic case, it is conjectured that the Chern marker characterizes the existence of localized Wannier basis for gapped systems \cite{2019MarcelliMonacoMoscolariPanati,2020MarcelliMoscolariPanati}. 
Before continuing to state the conjecture more precisely and state the main result of this paper, which confirms the conjecture in one direction, let us start by making some definitions:
\begin{definition}
\label{def:exp-loc-kern}
  Suppose that $A$ is a bounded linear operator on $L^2(\R^d) \rightarrow L^2(\R^d)$. We say that $A$ admits an \textit{exponentially localized kernel} with decay rate $\gamma$, if $A$ admits an integral kernel $A(\cdot, \cdot) : \R^d \times \R^d \rightarrow \C$ and there exists a finite, positive constant $C$ so that:
\[
    |A(\vec{x}, \vec{x}')| \leq C e^{-\gamma | \vec{x} - \vec{x}'|} \quad a.e.
\]
\end{definition}

\begin{definition}[$s$-localized generalized Wannier basis]
\label{def:s-loc}
Given an orthogonal projector $P$, we say an orthonormal basis $\{ \psi_{\alpha} \}_{\alpha \in \cI} \subseteq L^2(\R^2)$ is an \textit{$s$-localized generalized Wannier basis} for $P$ for some $s > 0$ if:
\begin{enumerate}
\item The collection $\{ \psi_{\alpha} \}_{\alpha \in \cI}$ spans $\range{(P)}$,
\item There exists a finite, positive constant $C$ and a collection of points $\{ \vec{\mu}_{\alpha} \}_{\alpha \in \cI} \subseteq \R^2$ such that for all $\alpha \in \cI$
\[
\int_{\R^2} \la \vec{x} - \vec{\mu}_{\alpha} \ra^{2s} |\psi_{\alpha}(\vec{x})|^2 \dd{\vec{x}} \leq C,
\]
where $\la \vec{x} - \vec{\mu}_{\alpha} \ra := ( \lvert\vec{x} - \vec{\mu}_{\alpha}\rvert^2 + 1)^{1/2}$ is the Japanese bracket. 
\end{enumerate}
We refer to the collection $\{ \vec{\mu}_{\alpha} \}_{\alpha \in \cI}$ as the \textit{center points} of the basis $\{ \psi_{\alpha} \}_{\alpha \in \cI}$.
\end{definition}
With these definitions, the localization dichotomy conjecture for non-periodic systems is as follows:
\begin{conjecture*}[Localization dichotomy for non-periodic gapped systems]
Let $P$ be an orthogonal projector which admits an exponentially localized kernel. Then the following statements are equivalent:
  \begin{enumerate}[label=(\alph*)]
  \item $P$ admits a generalized Wannier basis that is exponentially localized.
  \item $P$ admits a generalized Wannier basis that is $s$-localized for $s = 1$.
  \item $P$ is topologically trivial in the sense that its Chern marker $C(P)$ exists and is equal to zero.
  \end{enumerate}
\end{conjecture*}
Note that obviously (a) implies (b). For the other equivalence, there have been a few works devoted to the study of non-periodic localization dichotomy. In particular, recent work \cite{2020MarcelliMoscolariPanati} has shown that (b) $\Rightarrow$ (c) with $s > 5$. Additionally, our previous work \cite{2021LuStubbs} has shown that (b) $\Rightarrow$ (a) (and hence (b) $\Rightarrow$ (c)) with $s > 5/2$. In this paper, we improve upon these previous works by showing that (b) $\Rightarrow$ (c) for $s > 1$. Formally stated, the main result of this paper is the following:
\begin{theorem}
\label{thm:main}
Suppose that $P$ is an orthogonal projection on $L^2(\R^2)$ which admits an exponentially localized kernel. If $P$ admits an $(1 + \delta)$-localized generalized Wannier basis for some $\delta > 0$, then the Chern marker $C(P)$ vanishes.
\end{theorem}
Our proof of Theorem \ref{thm:main} can be extended to $s = 1$ by making some additional technical assumptions on the generalized Wannier basis however a full proof for the $s = 1$ case remains open. We note that Theorem~\ref{thm:main} establishes only one part of the localization dichotomy, while the other direction, $C(P) = 0$ implies the existence of localized generalized Wannier basis, is still quite open. 

\subsection*{Notations}
Vectors in $\R^d$ will be denoted by bold face with their components denoted by subscripts. For example, $\vec{v} = (v_1, v_2, v_3, \cdots, v_d) \in \R^d$. For any $\vec{v} \in \R^d$, we use $| \cdot |$ to denote its $\ell^2$-norm and $| \cdot |_{\infty}$ to denote its $\ell^\infty$-norm; that is, $| \vec{v} | := \bigl(\sum_{i=1}^d v_i^2\bigr)^{1/2}$, $| \vec{v} |_{\infty} := \max_{i} |v_i|$. For any $\vec{x} \in \R^2$ and $a \in \R^+$, we define $\chi_{a}$ to be the indicator function of the set $[-a, a)^2$ and $B_a(\vec{x})$ be the ball of radius $a$ centered at $\vec{x}$.

For any $f : \R^2 \rightarrow \C$, we will use $\|f\|$ to denote the $L^2$-norm. For any bounded linear operator $A$ on $L^2(\R^2)$, we adopt the following conventions:
\begin{itemize}
\item Let $\| A \|$ denote the spectral norm of $A$, $\| A \| := \sup_{\| f \| = 1} \| A f \|$.
\item If $A$ is compact, let $\{ \sigma_n(A) \}_{n = 1}^\infty$ denote the singular values of $A$ in decreasing order (i.e. if $i < j$ then $\sigma_i(A) \geq \sigma_j(A)$). 
\item If $A$ is compact, let $\| A \|_{\fS_p} = \bigl(\sum_{n=1}^\infty \sigma_n(A)^p \bigr)^{1/p}$ denote the Schatten $p$-norm for any $p \geq 1$. 
\end{itemize}
Note that with this convention $\| A \| = \| A \|_{\fS_\infty}$.

In our estimates, $C$ is used as a generic constants whose value may change from line to line. We also write $A \lesssim B$ if there exists a constant $C$ such that $A \leq C B$.

\subsection*{Organization}
The remainder of this paper is organized as follows. In Section \ref{sec:main-theorem}, we outline the proof of Theorem \ref{thm:main} relying on a number of propositions (Proposition \ref{prop:chiP-PL-approx}, \ref{prop:PL-chern}, and \ref{prop:p_x_pl_bd}). Next, in Section \ref{sec:technical-est} we state and prove three important technical estimates which are central to the proofs of these propositions. We provide proofs of Proposition \ref{prop:chiP-PL-approx} in Section \ref{sec:chiP-PL-approx}, Proposition \ref{prop:PL-chern} in Section \ref{sec:PL-chern}, and Proposition \ref{prop:p_x_pl_bd} in Section \ref{sec:p_x_pl_bd}, respectively.

\section{Proof of Main Theorem}
\label{sec:main-theorem}
We begin our proof by recalling the notion of bounded density which was introduced in \cite{2021LuStubbs} to simplify the analytic estimates. After recalling the consequences of bounded density (in particular, Lemma \ref{lem:fdc}), we will use these results to prove the main theorem.

\subsection{Bounded Density}
\label{sec:bdd-density}
We begin with the definition of bounded density
\begin{definition}
\label{def:bdd-density}
We say that a collection of points $\{ \vec{\mu}_{\alpha} \}_{\alpha \in \cI}$ has \textbf{bounded density} if there exists a constant $M < \infty$ such that for all $\vec{x} \in \R^2$ we have
\[
    \#  \{ \alpha :  \vec{\mu}_{\alpha} \in B_1(\vec{x}) \}  \leq M
\]  
\end{definition}
Importantly, if orthogonal projector $P$ has an exponentially localized kernel, one can show that the center points of every well localized basis must have bounded density. 
\begin{lemma}
  Let $P$ be an orthogonal projector which admits an exponentially localized kernel. If $\{ \psi_{\alpha} \}_{\alpha \in \cI}$ is an $s$-localized generalized Wannier basis for $P$ for some $s > 0$, then the center points for $\{ \psi_{\alpha} \}_{\alpha \in \cI}$ have bounded density.
\end{lemma}
\begin{proof}
For this proof, let $\chi_{B_r(\vec{a})}$ denote the characteristic function of the ball $B_r(\vec{a})$: $\chi_{B_r(\vec{a})}(\vec{x}) = 1$ if $\vec{x} \in B_r(\vec{a})$ and zero otherwise. We start by observing two important facts. 
\begin{enumerate}[label=(\roman*), wide]
    \item If $\{ \psi_{\alpha} \}_{\alpha \in \cI}$ is an $s$-localized basis for $s > 0$ with center points $\{ \vec{\mu}_\alpha \}_{\alpha \in \cI}$ then we have that
\begin{align*}
\| (1 - \chi_{B_r(\vec{\mu}_\alpha)}) \psi_{\alpha} \|^2  
& = \int_{\R^2} (1 - \chi_{B_r(\vec{\mu}_\alpha)}(\vec{x})) \frac{\la \vec{x} - \vec{\mu_\alpha} \ra^{2s}}{\la \vec{x} - \vec{\mu_\alpha} \ra^{2s}} |\psi_{\alpha}(\vec{x})|^2 \dd{\vec{x}} \\
& \lesssim r^{-2s}  \int_{\R^2}\la \vec{x} - \vec{\mu_\alpha} \ra^{2s} |\psi_{\alpha}(\vec{x})|^2\dd{\vec{x}} 
\end{align*}
Hence, since the collection $\{ \psi_{\alpha} \}_{\alpha \in \cI}$ is $s$-localized, there exists a constant $C$, uniform in $\alpha$, so that $\| (1 - \chi_{B_r(\vec{\mu}_\alpha)}) \psi_{\alpha} \|^2 \leq C r^{-2s}$. Thus we can find a radius $R > 0$ so that for all $\alpha \in \cI$ and all $r \geq R$
\begin{equation}
\label{eq:bdd-density-R}
\| (1 - \chi_{B_{r}(\vec{\mu}_\alpha)}) \psi_{\alpha} \|^2 \leq \frac{1}{2}.
\end{equation}
Since $\| (1 - \chi_{B_{r}(\vec{\mu}_\alpha)}) \psi_{\alpha} \|^2 + \| \chi_{B_{r}(\vec{\mu}_\alpha)} \psi_{\alpha} \|^2 = 1$, have that for all $r \geq R$, $\| \chi_{B_{r}(\vec{\mu}_\alpha)} \psi_{\alpha} \|^2 \geq \frac{1}{2}$.
\item Since $P$ admits an exponentially localized kernel, one easily checks that there exists a constant $K$ so that for all $a \in \R^2$: 
\begin{equation}
\label{eq:bdd-density-P}
\| \chi_{B_r(a)} P \|_{\fS_2}^2
= \int_{\R^2} \int_{\R^2} \chi_{B_r(a)}(\vec{x}) |P(\vec{x},\vec{y})|^2 \dd{\vec{x}} \dd{\vec{y}} \leq K r^2
\end{equation}
\end{enumerate}
Now let $\{ \psi_{\alpha} \}_{\alpha \in \cI}$ be an $s$-localized basis for some $s > 0$ and towards a contradiction suppose that the center points of this basis does not have bounded density. 

Since the center points for this basis do not have bounded density, we can find a point $\vec{x}^* \in \R^2$ so that the ball $B_1(\vec{x}^*)$ has more than $4 K (R + 1)^2$ center points where the constant $R$ is from Equation \eqref{eq:bdd-density-R} and the constant $K$ is from Equation \eqref{eq:bdd-density-P}. Let us denote the set of these center points by $\cA := \{ \alpha : \vec{\mu}_{\alpha} \in B_1(\vec{x}^*) \}$.

Due to Equation \eqref{eq:bdd-density-R} we have that
\[
\| \chi_{B_{R+1}(\vec{x}^*)} P \|_{\fS_2}^2 \leq K (R + 1)^2
\]
but on the other hand we have that
\begin{align*}
\| \chi_{B_{R+1}(\vec{x}^*)} P \|_{\fS_2}^2
& \geq \sum_{\alpha \in \cA} \| \chi_{B_{R+1}(\vec{x}^*)} \psi_{\alpha} \|^2 
 \geq \frac{1}{2} \big(\# \cA \big)
 \geq 2 K (R + 1)^2
\end{align*}
where we have used that $\alpha \in \cA$ implies that $B_R(\vec{\mu}_\alpha) \subseteq B_{R+1}(\vec{x}^*)$ and Equation \eqref{eq:bdd-density-R}. This is a contradiction and hence the center points of $\{ \psi_{\alpha} \}_{\alpha \in \cI}$ must have bounded density.
\end{proof}
The usefulness of the notion of bounded density is that we can effectively treat any basis with bounded density to have its center points on the integer lattice.
\begin{lemma}
\label{lem:fdc}
Let $\{ \psi_{\alpha} \}_{\alpha \in \cI}$ is a $s$-localized basis with center points $\{ \vec{\mu}_{\alpha} \}_{\alpha \in \cI}$. If we additionally assume that the center points have bounded density, then we may find a positive integer $M$ so that we can relabel the basis as $\{ \psi_{\vec{m}}^{(j)} \}$ where $\vec{m} \in \Z^2$ and $j \in \{ 1, \cdots, M \}$. Furthermore, the center point of $\psi_{\vec{m}}^{(j)}$ can be taken to be $\vec{m}$ without loss of generality.
\end{lemma}
\begin{proof}
For each $\vec{m} \in \Z^2$ let us define the unit square centered at $\vec{m}$ as follows
\[
S_{\vec{m}} := \bigg[ m_1 - \frac{1}{2}, m_1 + \frac{1}{2} \bigg) \times \bigg[ m_2 - \frac{1}{2}, m_2 + \frac{1}{2} \bigg).
\]
Since the basis $\{ \psi_{\alpha} \}_{\alpha \in \cI}$ has center points with bounded density, we know that there are at most $M$ center points contained in the square $S_{\vec{m}}$ (as it is contained in $B_1(\vec{m})$).  Because of this, we can relabel this basis as $\{ \psi_{\vec{m}}^{(j)} \}$ where $\psi_{\vec{m}}^{(j)}$ has its center in $S_{\vec{m}}$ and $j$ is a degeneracy index which runs from $\{ 1, \cdots, M \}$. If $S_{\vec{m}}$ has fewer than $M$ center points, say it has $j^*$, then we define $\psi_{\vec{m}}^{(j)} \equiv 0$ for all $j > j^*$. Strictly speaking this enlarged set is no longer a basis, but it does not really matter as we are only interested in the Chern marker, which is unchanged by this enlargement. 

If $\psi_{\vec{m}}^{(j)}$ initially had center point $\vec{\mu}_{\alpha}$, by construction $| \vec{m} - \vec{\mu}_{\alpha} |_2 \leq \frac{\sqrt{2}}{2}$. Therefore, using triangle inequality, it is easy to check that the collection $\{ \psi_{\vec{m}}^{(j)} \}$ is $s$-localized if we choose $\vec{m}$ as the center point of $\psi_{\vec{m}}^{(j)}$ instead. 
\end{proof}

Throughout our proof, we will assume that $M = 1$ to simplify notation. Considering the case $M > 1$ only has the effect of introducing a multiplicative factor of $M$ to some of our upper bounds and does not change the overall argument or results.

\subsection{Proof outline}
\label{sec:chern-marker}
As discussed in the previous section, as a consequence of bounded density (with $M = 1$), any $s$-localized basis may be written as $\{ \psi_{\vec{m}} \}$ where $\psi_{\vec{m}}$ has its center point at $\vec{m}$. Given a fixed choice of basis, we can now define the projector $P_{L}$ which projects onto the basis functions centered within the box of size $L$:
\begin{equation}
\label{eq:pl-def}
    P_L := \sum_{| \vec{m} |_{\infty} \leq L} \ket{\psi_{\vec{m}}} \bra{\psi_{\vec{m}}}.
\end{equation}
Throughout the rest of this paper, we will assume that projector $P_L$ is fixed and defined through a basis $\{ \psi_{\vec{m}} \}$ which is $(1 + \delta)$-localized for some $\delta > 0$.

Unlike $\chi_L P$ which appears in the definition of the Chern marker, the projector $P_L$ has finite rank and $\range{(P_L)} \subseteq \range{(P)}$. In some sense, the orthogonal projector $P_L$ captures the local information of $P$ in more controlled way than multiplying $P$ by the cutoff $\chi_L$ as in the definition of the Chern marker. Importantly, thanks to the decay property of the basis functions $\{ \psi_{\vec{m}} \}$, approximating $\chi_L P$ with $P_L$ incurs an error which is subleading compared to the area of $\chi_L$:
\begin{proposition}
\label{prop:chiP-PL-approx}
Suppose that $P$ admits a $(1 + \delta)$-localized basis where $\delta > 0$. There exists a constant $C$ such that for all $L \geq 1$:
\[
\| \chi_{L} P - P_{L} \|_{\fS_2} \leq C L^{2/3}
\]
\end{proposition}
\begin{proof}
Proven in Section \ref{sec:chiP-PL-approx}.
\end{proof}
As a consequence of this proposition, we can show that replacing $\chi_L P$ with $P_L$ in the definition of the Chern marker does not change the overall limit:
\begin{proposition}
\label{prop:PL-chern}
If $P$ admits a $(1+\delta)$-localized generalized Wannier basis where $\delta > 0$ then
  \begin{equation}
   \lim_{L \rightarrow \infty} \frac{1}{L^2} \left\| \chi_{L} P \Big[ [ X, P ], [ Y, P ] \Big] P \chi_{L} - P_L \Big[ [ X, P ], [ Y, P ] \Big] P_L \right\|_{\fS_1} = 0.
  \end{equation}
Hence
\begin{equation}
    \lim_{L \rightarrow \infty} \frac{2 \pi i}{4L^2} \tr{\left( \chi_{L} P \Big[ [ X, P ], [ Y, P ] \Big] P \chi_{L}\right)} = \lim_{L \rightarrow \infty} \frac{2 \pi i}{4 L^2} \tr{\left( P_L \Big[ [ X, P ], [ Y, P ] \Big] P_L\right)}
\end{equation}
whenever at least one of the above limits exists.
\end{proposition}
\begin{proof}
Proven in Section \ref{sec:PL-chern}
\end{proof}
Hence to prove Theorem \ref{thm:main} it suffices to show that if $P$ admits an $(1 + \delta)$-localized generalized Wannier basis then
\begin{equation}
\label{eq:pl-chern-marker}
\lim_{L \rightarrow \infty} \frac{2 \pi i}{4 L^2} \tr{\left( P_L \Big[ [ X, P ], [ Y, P ] \Big] P_L\right)} = 0.
\end{equation}
Towards proving Equation \eqref{eq:pl-chern-marker}, we begin by observing that since $P_L$ is defined through a $(1 + \delta)$-localized basis, the position operator $X$ is a bounded operator on $\range{(P_L)}$ for each $L$. In particular, we have that
\begin{align*}
\| X P_L \|^2 
& \leq \sum_{| \vec{m} |_{\infty} \leq L} \| X \psi_{\vec{m}} \|^2 \\
& \leq \sum_{| \vec{m} |_{\infty} \leq L} \Big( \| (X - m_1) \psi_{\vec{m}} \| + |m_1| \| \psi_{\vec{m}} \|  \Big)^2 \\
& \leq  \sum_{| \vec{m} |_{\infty} \leq L} \Big( \| (X - m_1) \psi_{\vec{m}} \| + L  \Big)^2 \\
& \lesssim L^4
\end{align*}
Similarly, it is easily checked that $Y$ is also a bounded operator on $\range{(P_L)}$.

We will now use the fact that $X$ and $Y$ are both bounded operators on $\range{(P_L)}$ to perform some algebraic manipulations. Using the fact that $P^2 = P$ and $[X, Y] = 0$, one can  verify that (see also \cite{2020MarcelliMoscolariPanati,2020StubbsWatsonLu1})
\[
P \Big[ [ X, P ], [ Y, P ] \Big] P = [ PXP, PYP ].
\]
Therefore, since $P_L = P_L P = P P_L$, we have the following:
\begin{align*}
    P_L \Big[ [ X, P ], [ Y, P ] \Big] P_L 
    & = P_L [ PXP, PYP ] P_L \\
    & = P_L X P Y P_L - P_L Y P X P_L \\
    & = P_L X (P - P_L + P_L) Y P_L - P_L Y (P - P_L + P_L) X P_L \\
    & = [P_L X P_L, P_L Y P_L] + P_L X (P - P_L) Y P_L - P_L Y (P - P_L) X P_L.
\end{align*}
These manipulations are justified since $X$ and $Y$ are bounded operators on $\range{(P_L)}$. Since $P_L$ is finite rank, $[P_L X P_L, P_L Y P_L]$ is traceless and hence
\begin{equation}
\label{eq:PL-chern-equiv}
\tr \left( P_L \Big[ [ X, P ], [ Y, P ] \Big] P_L \right) = \tr\Bigl( P_L X (P - P_L) Y P_L - P_L Y (P - P_L) X P_L \Bigr).
\end{equation}
Hence using H{\"o}lder's inequality and $(P - P_L) = (P - P_L)^2$, we have that
\begin{align*}
|\tr \left( P_L \Big[ [ X, P ], [ Y, P ] \Big] P_L \right)|
& \leq \| P_L X (P - P_L) Y P_L \|_{\fS_1} + \| P_L Y (P - P_L) X P_L \|_{\fS_1} \\
& \leq 2 \| (P - P_L) X P_L \|_{\fS_2} \| (P - P_L) Y P_L \|_{\fS_2}. \numberthis{} \label{eq:pl_chern_upper}
\end{align*}

\begin{proposition}
\label{prop:p_x_pl_bd}
If $P$ admits a $(1 + \delta)$-localized generalized Wannier basis where $\delta > 0$ then
\begin{align}
& \lim_{L \rightarrow \infty} \frac{1}{L^2} \| (P - P_L) X P_L \|_{\fS_2}^2 = 0 \\[1ex]
& \lim_{L \rightarrow \infty} \frac{1}{L^2} \| (P - P_L) Y P_L \|_{\fS_2}^2 = 0.
\end{align}
\end{proposition}
Since the mapping $x \mapsto \sqrt{x}$ is continuous for $x > 0$, Proposition \ref{prop:p_x_pl_bd} and Equation \eqref{eq:pl_chern_upper} imply that
\[
\lim_{L \rightarrow \infty} \frac{1}{L^2} \Bigl\lvert\tr \Bigl( P_L \Big[ [ X, P ], [ Y, P ] \Big] P_L \Bigr)\Bigr\rvert = 0
\]
which proves Equation \eqref{eq:pl-chern-marker}, completing the proof of Theorem \ref{thm:main}.

\section{Technical Estimates}
\label{sec:technical-est}
In this section we prove two technical estimates (Proposition \ref{prop:near-bd} and Proposition \ref{prop:far-bd}) which are fundamental in our proofs of Proposition \ref{prop:chiP-PL-approx} and Proposition \ref{prop:p_x_pl_bd}.

\begin{proposition}
\label{prop:near-bd}
If $P$ admits a $(1 + \delta)$-localized generalized Wannier basis then for all $a, b \geq 1$: 
\[
\| (1 - \chi_{a+b}) P_{a} \|_{\fS_2}^2 \lesssim a^2 b^{-2(1+\delta)}
\]
\end{proposition}

\begin{proof}
We can expand the Hilbert-Schmidt norm we want to bound as follows:
\begin{equation}
\label{eq:near-bd1}
\| (1 - \chi_{a+b}) P_{a} \|_{\fS_2}^2 = \sum_{|\vec{m}|_\infty \leq a} \| (1 - \chi_{a+b}) \psi_{\vec{m}} \|^2.
\end{equation}
Because of the separation between the sets $\{ \vec{m} \in \Z^2 : | \vec{m} |_{\infty} \leq a \}$ and $\supp{(1 - \chi_{a+b})}$ we can show that each of the terms in the above sum are small. In particular
\begin{align*}
    \| (1 - \chi_{a + b}) \psi_{\vec{m}} \|^2
    & = \int_{\R^2} (1 - \chi_{a+b}(\vec{x})) |  \psi_{\vec{m}}(\vec{x})|^2 \dd{\vec{x}} \\
    & = \int_{\R^2} (1 - \chi_{a+b}(\vec{x})) \frac{(1 + |x_1 - m_1| + |x_2 - m_2| )^{2(1+\delta)}}{(1 + |x_1 - m_1| + |x_2 - m_2| )^{2(1+\delta)}} |  \psi_{\vec{m}}(\vec{x})|^2 \dd{\vec{x}} 
\end{align*}
Since $| \vec{m} |_{\infty} \leq a$ we have the pointwise bound
\[
\frac{(1 - \chi_{a+b}(\vec{x}))}{(1 + |x_1 - m_1| + |x_2 - m_2| )} \leq \frac{1}{1 + (a+b) - | \vec{m} |_{\infty}} \leq \frac{1}{1 + b}.
\]
Therefore, for each since $\psi_{\vec{m}}$ is $(1 + \delta)$-localized we have that:
\begin{align*}
\| (1 - & \chi_{a + b}) \psi_{\vec{m}} \|^2 \\
& \leq (1 + b)^{-2(1+\delta)} \int_{\R^2} (1 + |x_1 - m_1| + |x_2 - m_2| )^{2(1+\delta)} |  \psi_{\vec{m}}(\vec{x})|^2 \dd{\vec{x}}  \\
& \leq C b^{-2(1+\delta)}
\end{align*}
for some absolute constant $C$. Using this bound in Equation \eqref{eq:near-bd1}, we conclude that
\begin{align*}
\| (1 - \chi_{a+b}) P_{a} \|_{\fS_2}^2
& \leq \sum_{|\vec{m}|_\infty \leq a} C b^{-2(1+\delta)} \\[1ex]
& \lesssim a^2 b^{-2(1+\delta)}
\end{align*}
which completes the proof.
\end{proof}

\begin{proposition}
\label{prop:far-bd}
If $P$ admits a $(1+\delta)$-localized generalized Wannier basis then for all $a, b \geq 1$: 
\[
\| \chi_{a} ( P - P_{a + b} ) \|_{\fS_2}^{2} \lesssim  b^{-\delta} + a b^{-(1+\delta)}
\]
\end{proposition}
We start by stating a lemma which we prove at the end of the section. 
\begin{lemma}
\label{lem:decay-trick}
Suppose that $\{ \psi_{\vec{m}} \}$ is a $(1+\delta)$-localized basis. For any $a \geq 1$ we have the following bounds depending on the location of $\vec{m}$ in relation to $\supp{(\chi_{a})}$:
\begin{enumerate}[label=(\roman*)]
\item If $|m_1| > a$ and $|m_2| > a$ then
\[
\| \chi_{a} \psi_{\vec{m}} \|^2 \lesssim \la |m_1| - a \ra^{-(1 + \delta)} \la |m_2| - a \ra^{-(1 + \delta)}.
\]
\item If $|m_1| > a$ and $|m_2| \leq a$ then
\[
\| \chi_{a} \psi_{\vec{m}} \|^2 \lesssim \la |m_1| - a \ra^{-2(1 + \delta)}
\]
\item If $|m_1| \leq a$ and $|m_2| > a$ then
\[
\| \chi_{a} \psi_{\vec{m}} \|^2 \lesssim \la |m_2| - a \ra^{-2(1 + \delta)}
\]
\end{enumerate}
\end{lemma}
With this lemma in hand, we can now prove Proposition \ref{prop:far-bd}.
\begin{proof}[Proof of Proposition \ref{prop:far-bd}]
By the properties of the Hilbert-Schmidt norm we see that
\begin{align*}
\| \chi_{a} ( P - P_{a + b} ) \|_{\fS_2}^{2}
& = \sum_{\| \vec{m} \| > a + b} \| \chi_{a} \psi_{\vec{m}} \|^2
\end{align*}
We now split the set $\{ \vec{m} \in \Z^2 : | \vec{m} |_{\infty} > a + b \}$ into three parts and bound each part separately
\[
\begin{split}
S_1 : = \Big\{ \vec{m} : |m_1| > a + b \text{ and } |m_2| > a + b \Big\}  \\[1ex]
S_2 : = \Big\{ \vec{m} : |m_1| > a + b \text{ and } |m_2| \leq a + b \Big\}  \\[1ex]
S_3 : = \Big\{ \vec{m} : |m_1| \leq a + b \text{ and } |m_2| > a + b \Big\}  \\
\end{split}
\]

We start with controlling $S_1$, by applying Lemma \ref{lem:decay-trick}(1) we have that
\begin{align*}
    \sum_{\vec{m} \in S_1} \| \chi_{a} \psi_{\vec{m}} \|^{2} 
    & \leq \sum_{\vec{m} \in S_1} \frac{C}{\la |m_1| - a \ra^{-(1 + \delta)} \la |m_2| - a \ra^{-(1 + \delta)}} \\
    & \leq C b^{-\delta} \sum_{\vec{m} \in S_1} \frac{1}{\la |m_1| - a \ra^{(1 + \delta/2)} \la |m_2| - a \ra^{(1 + \delta/2)}} \\
    & \leq C b^{-\delta} \sum_{\vec{m} \in \Z^2} \frac{1}{\la |m_1| - a \ra^{(1 + \delta/2)} \la |m_2| - a \ra^{(1 + \delta/2)}} 
\end{align*}
where in the second to last line we have used that since $\vec{m} \in S_1$, $\min\{ \la |m_1| - a \ra, \la |m_2| - a \ra \} > b$. Therefore, 
\[
\sum_{\vec{m} \in S_1} \| \chi_{a} \psi_{\vec{m}} \|^{2} \lesssim b^{-\delta}
\]

We now turn to bound the sum for $\vec{m} \in S_2$. Applying Lemma \ref{lem:decay-trick}(2) we have that there exists a constant $C$ such that
\begin{align*}
\sum_{\vec{m} \in S_2} \| \chi_{a} \psi_{\vec{m}} \|^{2}
& \leq \, \sum_{|m_1| > a + b} \sum_{|m_2| \leq a + b} \frac{C}{\la |m_1| - a \ra^{2(1 + \delta)}} \\
& \leq C b^{-(1 + \delta)} \sum_{|m_1| > a + b} \sum_{|m_2| \leq a + b} \frac{1}{\la |m_1| - a \ra^{(1 + \delta)}} \\
& \leq 2 C (a + b) b^{-(1 + \delta)} \sum_{m_1 \in \Z} \frac{1}{\la |m_1| - a \ra^{(1 + \delta)}} 
\end{align*}
where in the second to last line we have used that $\la |m_1| - a \ra > b$. Therefore, 
\[
\sum_{\vec{m} \in S_2} \| \chi_{a} \psi_{\vec{m}} \|^{2} \lesssim (a + b) b^{-(1 + \delta)}
\]

Repeating the same calculation for $S_3$ making the obvious changes we have that
\[
\sum_{\vec{m} \in S_3} \| \chi_{a} \psi_{\vec{m}} \|^{2} \lesssim (a + b) b^{-(1 + \delta)}
\]
Hence
\[
\| \chi_{a} ( P - P_{a + b} ) \|_{\fS_2}^{2} \leq C_1 b^{-\delta} + C_2 (a + b) b^{-(1 + \delta)} + C_3 (a + b) b^{-(1 + \delta)}
\]
which proves the result.
\end{proof}

It remains to prove Lemma \ref{lem:decay-trick} to finish the proof. 
\begin{proof}[Proof of Lemma \ref{lem:decay-trick}]
We will focus on the case when $|m_1| > a$ and $|m_2| > a$ and note the changes which must be made for the other cases. For these estimates, we will introduce the strip characteristic functions $\chi^{\text{strip},X}_{D}$ and $\chi^{\text{strip},Y}_{D}$ defined as follows
\[
\chi^{\text{strip},X}_{D}(\vec{x}) = 
\begin{cases}
1 & |x_1| \leq D \\
0 & \text{otherwise}
\end{cases}
\qquad \qquad
\chi^{\text{strip},Y}_{D}(\vec{x}) = 
\begin{cases}
1 & |x_2| \leq D \\
0 & \text{otherwise}
\end{cases}
\]
Next, let us define the distances $D_x := |m_1| - a$ and $D_y := |m_2| - a$. With these definitions, it is clear that up to a set of measure zero:
\[
\chi^{\text{strip},X}_{D_x}(\vec{x} - \vec{m}) \chi_{a}(\vec{x}) = 0 \quad \text{and} \quad \chi^{\text{strip},Y}_{D_y}(\vec{x} - \vec{m}) \chi_{a}(\vec{x}) = 0
\]
Therefore,
\begin{align*}
\| \chi_L \psi_{\vec{m}}  \|^2
    & = \int_{\R^2} \chi_{a}(\vec{x}) | \psi_{\vec{m}} (\vec{x})|^2 \dd{\vec{x}} \\
    & = \int_{\R^2} \chi_{a}(\vec{x}) \Big(1 - \chi^{\text{strip},X}_{D_x}(\vec{x} - \vec{m})\Big) | \psi_{\vec{m}} (\vec{x})|^2 \dd{\vec{x}} \\
    & = \int_{\R^2} \chi_{a}(\vec{x}) \Big(1 - \chi^{\text{strip},X}_{D_x}(\vec{x} - \vec{m})\Big) \frac{\la x_1 - m_1\ra^{(1+\delta)}}{\la x_1 - m_1\ra^{(1+\delta)}} | \psi_{\vec{m}} (\vec{x})|^2 \dd{\vec{x}}
\end{align*}
By definition of $\chi^{\text{strip},X}_{D_x}$ we have the pointwise bound:
\[
\frac{1 - \chi^{\text{strip},X}_{D_x}(\vec{x} - \vec{m})}{\la x_1 - m_1\ra} \leq \frac{1}{\la |m_1| - a \ra}
\]
Therefore,
\[
\| \chi_L \psi_{\vec{m}}  \|^2 \leq \frac{1}{\la |m_1| - a \ra^{(1+\delta)}}\int_{\R^2} \chi_{a}(\vec{x}) \la x_1 - m_1\ra^{(1+\delta)} | \psi_{\vec{m}} (\vec{x})|^2 \dd{\vec{x}}.
\]
By similar logic
\begin{align*}
\int_{\R^2} & \chi_{a}(\vec{x}) \la x_1 - m_1\ra | \psi_{\vec{m}} (\vec{x})|^2 \dd{\vec{x}} \\
& = \int_{\R^2} \chi_{a}(\vec{x}) \Big(1 - \chi^{\text{strip},Y}_{D_y}(\vec{x} - \vec{m})\Big) \frac{\la x_2 - m_2\ra^{(1+\delta)}}{\la x_2 - m_2\ra^{(1+\delta)}} \la x_1 - m_1\ra | \psi_{\vec{m}} (\vec{x})|^2 \dd{\vec{x}} \\
& \leq \frac{1}{\la |m_2| - a \ra^{(1+\delta)}} \int_{\R^2} \chi_{a}(\vec{x}) \la x_1 - m_1\ra^{(1+\delta)} \la x_2 - m_2\ra^{(1+\delta)} | \psi_{\vec{m}} (\vec{x})|^2 \dd{\vec{x}}.
\end{align*}
Hence
\[
\| \chi_a \psi_{\vec{m}}  \|^2 \leq \frac{1}{\la |m_1| - a \ra^{(1+\delta)} \la |m_2| - a \ra^{(1+\delta)}}\int_{\R^2} \chi_{a}(\vec{x}) \la x_1 - m_1\ra^{(1+\delta)} \la x_2 - m_2\ra^{(1+\delta)} | \psi_{\vec{m}} (\vec{x})|^2 \dd{\vec{x}}.
\]
Now recall that the geometric mean is bounded by the arithmetic mean so
\[
\la x_1 - m_1\ra^{(1+\delta)} \la x_2 - m_2\ra^{(1+\delta)} \leq \frac{1}{2} \Big(\la x_1 - m_1\ra^{2(1+\delta)} + \la x_2 - m_2\ra^{2(1+\delta)} \Big)
\]
Therefore,
\[
\| \chi_a \psi_{\vec{m}}  \|^2 \leq \frac{\| \la X - m_1 \ra^{(1+\delta)} \psi_{\vec{m}}  \|^2 + \| \la Y - m_2 \ra^{(1+\delta)} \psi_{\vec{m}}  \|^2}{2\la |m_1| - a \ra^{(1+\delta)} \la |m_1| - a \ra^{(1+\delta)}}.
\]
which implies the result since $\psi_{\vec{m}}$ is $(1+\delta)$-localized.

The case $|m_1| > a$ and $| m_2 | \leq a$ follows by inserting $\la x_1 - m_1 \ra^{2(1+\delta)} \la x_1 - m_1 \ra^{-2(1+\delta)}$ instead of $\la x_1 - m_1 \ra^{(1+\delta)} \la x_2 - m_2 \ra^{(1+\delta)} \la x_1 - m_1 \ra^{-(1+\delta)} \la x_2 - m_2 \ra^{-(1+\delta)}$; the case $|m_1| \leq a$ and $| m_2 | > a$ follows similarly.
\end{proof}

\section{Proof of Proposition \ref{prop:chiP-PL-approx}}
\label{sec:chiP-PL-approx}
Let us start by fixing some $\ell$ where $\ell \in [1, L)$ to be chosen later. We can split the quantity we would like to bound into four parts: 
\begin{align*}
    \| \chi_{L} P - P_L \|_{\fS_2} 
    & \leq \| \chi_{L} (P - P_L) \|_{\fS_2} + \| (1 -\chi_{L}) P_L \|_{\fS_2}\\
    & \leq \| \chi_{L} (P - P_{L + \ell}) \|_{\fS_2} +  \| \chi_{L} (P_{L + \ell} - P_L) \|_{\fS_2}\\
    & \hspace{5ex} + \| (1 -\chi_{L}) (P_L - P_{L - \ell}) \|_{\fS_2} + \| (1 -\chi_{L}) P_{L - \ell} \|_{\fS_2}
\end{align*}
The first term is bounded by Proposition \ref{prop:far-bd} by letting $a = L$, $b = \ell$ there exists a constant $C_1$ so that:
\begin{equation}
\label{eq:lem-bd1}
\| \chi_{L} (P - P_{L + \ell}) \|_{\fS_2} \leq C_1 (\ell^{-\delta} + L \ell^{-(1+\delta)})^{1/2}.
\end{equation}
The next two terms are bounded by observing that
\[
\begin{split}
&\rank{(P_{L+\ell} - P_{L})} \leq 4 ( (L+\ell)^2 - L^2) \leq 12 L \ell \\
& \rank{(P_{L} - P_{L-\ell})} \leq 4 (L^2 - (L-\ell)^2) \leq 12 L \ell 
\end{split}
\]
where we have used that $\ell < L$. Hence, there exists a constant $C_2$ so that
\begin{equation}
\label{eq:lem-bd2}
\| \chi_{L} (P_{L + \ell} - P_L) \|_{\fS_2} + \| (1 -\chi_{L}) (P_L - P_{L - \ell}) \|_{\fS_2} \leq C_2 (L \ell)^{1/2}.
\end{equation}
As for the final term, we can apply Proposition \ref{prop:near-bd} with $a = L - \ell$, $b = \ell$ to conclude that there exists a constant $C_3$ so that
\begin{equation}
\label{eq:lem-bd3}
\| (1 -\chi_{L}) P_{L - \ell} \|_{\fS_2} \leq C_3 (L^{2} \ell^{-2(1+\delta)})^{1/2}
\end{equation}
Combining the bounds in Equations \eqref{eq:lem-bd1}, \eqref{eq:lem-bd2}, \eqref{eq:lem-bd3}, we have that
\[
\| \chi_{L} P - P_L \|_{\fS_2} \leq C_1 (\ell^{-\delta} + L \ell^{-(1+\delta)})^{1/2} + C_2 (L \ell)^{1/2} + C_3 (L^{2} \ell^{-2(1+\delta)})^{1/2}.
\]
Now in the above equation we have four different terms each which have different big--O as $L \rightarrow \infty$:
\[
O(\ell^{-\delta}) \qquad O(L \ell^{-(1 + \delta)}) \qquad O(L \ell) \qquad O(L^2 \ell^{-2(1 + \delta)})
\]
Since $\ell > 1$, it's clear that the two dominating terms are $O(L \ell)$ and $O(L^2 \ell^{-2 (1 + \delta)})$. Since we are free to choose $\ell$, we will make a choice of $\ell$ so that these two terms balance. A simple calculation shows that choosing $\ell = L^{1/(3+2\delta)}$ gives
\begin{equation*}
\begin{split}
& L \ell = L^{2(2+\delta)/(3+2\delta)} \\
& L^2 \ell^{-2(1+\delta)} = L^{2(2+\delta)/(3+2\delta)} \\
& L \ell^{-(1 + \delta)} = L^{(2+\delta)/(3+2\delta)} 
\end{split}
\end{equation*}
This valid choice for $\ell$ since $\delta > 0$ so $\frac{1}{3+2\delta} < 1$. With this choice of $\ell$, we have that
\begin{align*}
\| \chi_{L} (P - P_L) \|_{2} \leq C_1 (L^{-\delta/(3+2\delta)} + 2 L^{(2+\delta)/(3+2\delta)})^{1/2} + C_2 L^{(2+\delta)/(3+2\delta)} + C_3 L^{(2+\delta)/(3+2\delta)}
\end{align*}
Hence, for $L \geq 1$
\[
\| \chi_{L} P - P_L \|_{\fS_2} \lesssim L^{(2+\delta)/(3+2\delta)}.
\]
The proof is completed by observing that for all $\delta \geq 0$:
\[
\frac{2 + \delta}{3 + 2 \delta} \leq \frac{2}{3}.
\]

\section{Proof of Proposition \ref{prop:PL-chern}}
\label{sec:PL-chern}
For this proof, let us abbreviate the commutator in the definition of the Chern marker as $\cC$, that is:
\[
\cC := \Big[ [ X, P ], [ Y, P ] \Big].
\]
With this notation, we have that:
\begin{align*}
    \chi_{L} P \cC P \chi_{L} - P_L \cC P_L
    & = (\chi_{L} P - P_L) \cC P \chi_{L} + P_L \cC (P \chi_{L} - P_L)
\end{align*}
Applying H{\"o}lder's inequality to the trace norm we want to bound, we have that
\begin{align*}
\| \chi_{L} P & \cC P \chi_{L} - P_L \cC P_L \|_{\fS_1} \\
& \leq \| (\chi_{L} P - P_L) \cC P \chi_{L} \|_{\fS_1} + \| P_L \cC (P \chi_{L} - P_L) \|_{\fS_1} \\
& \leq \|\chi_{L} P - P_L \|_{\fS_2} \| \cC \|_{\fS_{\infty}} \| P \chi_{L} \|_{\fS_2} + \| P_L \|_{\fS_2} \| \cC \|_{\fS_{\infty}}  \| P \chi_{L} - P_L \|_{\fS_2}.
\end{align*}
The right hand side can be upper bounded by observing that
\begin{enumerate}[label=(\roman*)]
\item Since $P$ admits an exponentially localized kernel, 
\[
\| \cC \|_{\fS_{\infty}} = \| \cC \| = \| [[X,P], [Y,P]] \| \leq 2 \| [X, P] \| \| [Y, P] \| \lesssim 1.
\]
\item Additionally, since $P$ admits an exponentially localized kernel, one easily checks that
\[
\| P \chi_{L} \|_{\fS_2} \lesssim L.
\]
\item Since $\rank{(P_L)} \leq 4 L^2$ and $\| P_L \| \leq 1$ we have that
\[
\| P_{L} \|_{\fS_2} \leq 2L.
\]
\item Proposition \ref{prop:chiP-PL-approx} implies that
\[
\| \chi_{L} P - P_L \|_{\fS_2} \lesssim L^{2/3}
\]
\end{enumerate}
Combining these four bounds, we conclude that
\[
\| \chi_{L} P \cC P \chi_{L} - P_L \cC P_L \|_{\fS_1} \lesssim L^{5/3}
\]
Hence
\[
\lim_{L \rightarrow \infty} \frac{1}{L^2} \| \chi_{L} P \cC P \chi_{L} - P_L \cC P_L \|_{\fS_1} = 0
\]
and the proposition is proved.

\section{Proof of Proposition \ref{prop:p_x_pl_bd}}
\label{sec:p_x_pl_bd}
Our main goal in this section is to show that the following quantity is $o(L^2)$:
\[
\| (P - P_L) X P_L \|_{\fS_2}^2 
\]
The corresponding bound for $Y$ follows by an analogous argument. 

Similar to the proof of Proposition \ref{prop:PL-chern}, our first step will be to introduce a length parameter $\ell \in [1, \frac{1}{2} L)$ to be fixed later. For any such choice of $\ell$, by the properties of the Hilbert-Schmidt norm we have that:
\begin{equation}
\label{eq:p_x_pl_bd1}
\| (P - P_L) X P_L \|_{\fS_2}^2 = \| (P - P_L) X P_{L - 2\ell} \|_{\fS_2}^2 + \| (P - P_L) X (P_L - P_{L - 2\ell}) \|_{\fS_2}^2
\end{equation}
The second of these terms can be shown to be $O(L\ell)$ using that $(P - P_L) (P_L - P_{L - 2\ell}) = 0$. In particular:
\begin{align*}
\| (P - P_L) X (P_L - P_{L - \ell}) \|_{\fS_2}^2
& = \sum_{L - 2\ell < |\vec{m}|_\infty \leq L} \| (P - P_L) X \psi_{\vec{m}} \|^2 \\
& = \sum_{L - 2\ell < |\vec{m}|_\infty \leq L} \|  (P - P_L) (X - m_1) \psi_{\vec{m}} \|^2 \\
& \leq \left( \sup_{\vec{m}} \| (X - m_1) \psi_{\vec{m}} \|^2 \right) \sum_{L - 2\ell < |\vec{m}|_\infty \leq L} 1 \\
& \lesssim L \ell.
\end{align*}
Therefore, this term is $o(L^2)$ so long as we choose $\ell = o(L)$. 

Returning to the first term in Equation \eqref{eq:p_x_pl_bd1}, using that $(P - P_L) P_{L - 2\ell} = 0$ we have that
\begin{align*}
\| (P - P_L) X P_{L - 2\ell} \|_{\fS_2}^2
& = \sum_{|\vec{m}|_\infty \leq L - 2\ell} \| (P - P_L) (X - m_1) \psi_{\vec{m}} \|^2
\end{align*}
We will now show that each of the terms in the above sum are small using Proposition \ref{prop:far-bd}. In particular, we have the following easy lemma
\begin{lemma}
\label{lem:p_x_pl_bd_lem}
If $\psi_{\vec{m}}$ is $(1+\delta)$-localized with center point $\vec{m}$ where $|\vec{m}|_{\infty} \leq L - 2\ell$, then
\[
\| (P - P_L) (X - m_1) \psi_{\vec{m}} \| \lesssim \ell^{-\delta} + L \ell^{-(1+\delta)}
\]
\end{lemma}
\begin{proof}
We start by inserting $\chi_{L - \ell} + (1 -\chi_{L - \ell})$ and applying triangle inequality
\begin{align*}
\| (P - P_L) & (X - m_1) \psi_{\vec{m}} \| \\
& \leq \| (P - P_L) \chi_{L - \ell} (X - m_1) \psi_{\vec{m}} \| +  \| (P - P_L) (1 -\chi_{L - \ell}) (X - m_1) \psi_{\vec{m}} \|
\end{align*}
By Proposition \ref{prop:far-bd} we have that
\[
\| (P - P_L) \chi_{L - \ell} (X - m_1) \psi_{\vec{m}} \| \leq \| (P - P_L) \chi_{L - \ell} \| \| (X - m_1) \psi_{\vec{m}} \| \lesssim \ell^{-\delta} + L \ell^{-(1+\delta)}.
\]
By repeating a similar argument as used in the proof of Proposition \ref{prop:near-bd} it is easily verified that
\begin{align*}
\| (P - P_L) (1 -\chi_{L - \ell}) (X - m_1) \psi_{\vec{m}} \| \leq  \| (1 -\chi_{L - \ell}) (X - m_1) \psi_{\vec{m}} \| \lesssim \ell^{-\delta}
\end{align*}
which proves the lemma.
\end{proof}
Using Lemma \ref{lem:p_x_pl_bd_lem} and it follows that:
\[
\| (P - P_L) X P_{L - \ell} \|_{\fS_2}^2 \lesssim L^2 (\ell^{-\delta} + L \ell^{-(1+\delta)})^{2}.
\]
For our final step, we will choose $\ell = L^{2/(2+\delta)}$. Since $\delta > 0$, note that this choice of $\ell$ is $o(L)$ consistent with our previous requirement. Furthermore, we have that
\[
\| (P - P_L) X P_{L - \ell} \|_{\fS_2}^2 \lesssim L^2 (L^{-2\delta/(2+\delta)} + L^{-\delta/(2+\delta)})^2.
\]
Since $\delta > 0$, we conclude that
\[
\lim_{L \rightarrow \infty} \frac{1}{L^2} \| (P - P_L) X P_{L - \ell} \|_{\fS_2}^2 = 0
\]
which completes the proof.
\bibliographystyle{plain}
\bibliography{topological}

\end{document}